\documentclass[12pt]{article}

\usepackage[usenames,dvipsnames]{pstricks}
\usepackage{epsfig}
\usepackage{pst-grad} 
\usepackage{pst-plot} 

\usepackage{amsmath} 
\usepackage{amsthm}
\usepackage{amssymb} 
\usepackage{graphicx} 
\usepackage{multicol}
\usepackage{url}
\usepackage[labelfont=it,textfont=it, small]{caption}

\usepackage[small,compact]{titlesec}
\newtheorem{theorem}{Theorem}[section]

\newtheorem{definition}{Definition}

\title{Manipulation Can Be Hard in Tractable Voting Systems Even for Constant-Sized Coalitions}

\author{Curtis Menton and Preetjot Singh\footnote{Supported in part by grants NSF-CCF-0915792 and NSF-CCF-1101479.} \\ Department of
    Computer Science \\ University of Rochester \\ Rochester, NY 14627
    USA}

\begin{document}
 \maketitle

\begin{abstract}

Voting theory has become increasingly integrated with computational
social choice and multiagent systems. Computational complexity has
been extensively used as a shield against manipulation of voting
systems, however for several voting schemes this complexity may cause
calculating the winner to be computationally difficult. Of the many
voting systems that have been studied with regard to election
manipulation, a few have been found to have an unweighted coalitional
manipulation problem that is NP-hard for a constant number of
manipulators despite having a winner problem that is in P\@. We survey
this interesting class of voting systems and the work that has
analyzed their complexity.
\end{abstract}

\section{Introduction}
Research in voting theory has become increasingly important due to the
ubiquity of voting systems.  While voting is most typically used for
political or organizational elections, recently it has become
relevant as a tool in multiagent systems and distributed
artificial intelligence.  From recommender
systems~\cite{gho-mun-her-sen:c:voting-for-movies,
  gil-hor-pen:c:collaborative-filtering} such as those seen in
Netflix which makes recommendations based on user activity, to
consensus mechanisms for planning in artificial
intelligence~\cite{eph-ros:cOUTbyJOURNAL:clarke-tax} and search engine
and metasearch engine design~\cite{Lifantsev:votingmodel,
  dwo-kum-nao-siv:c:rank-aggregation}, mechanisms that aggregate
individual `votes' have far-reaching applications.
Many fields of study in computer science such as mechanism design and
algorithmic game theory have become intertwined with research in
voting theory. In this paper we concentrate on results in
manipulation, that is, strategically changing one's vote so as to
change the result of the election.  

Manipulation, as opposed to
other ways of influencing election results, does not change
the structure of the election (as in control~\cite{bar-tov-tri:j:control} and cloning~\cite{springerlink:10.1007/BF00433944,elk-fal-sli:c:clone}) or
involve an external actor bribing voters to change their preferences
(as in bribery~\cite{fal-hem-hem:j:bribery} and campaign management~\cite{elk-fal:j:approx-camp,sch-fal-elk:c:camp}).
Manipulation does not require going outside of the election model but merely
involves voters picking their optimal votes.  Thus manipulation is the
most immediate and frequently relevant of these problems.


The most representative subproblem of manipulation is unweighted
coalitional manipulation (UCM). The model is simpler than one with
weighted voters and so the focus is squarely on the voting rule
itself, not on the interplay of differently weighted voters
forming a coalition. As such, results in UCM are a purer test of a
voting system's vulnerability to manipulation.

We will first explore the history and key results of voting theory that
imply significant issues with all voting systems, and subsequently we
will show how we can cope with some of these difficulties through complexity theory.

\subsection{Voting Theory}




Voting has long been used as a tool for collaborative decision making,
with democratic government known to have existed at least as far back
as 6th century BCE in ancient Greece.  For nearly as long people have
studied voting in an attempt to find the best election methods and solve
problems related to voting.


%
%

One milestone in the study of voting is the work of 13th century
mystic Ramon Llull.  Llull, a prominent Franciscan, was involved with
the Catholic church and researched methods for electing church
officials. 
Among his extensive body of work, encompassing at least 265 titles
on subjects ranging from controversial theological viewpoints to
romantic fiction, Llull's contributions to voting theory
stem from three works: \textit{Artifitium Electionis Personarum},
\textit{En qual manera Natana fo eleta a abadessa} (Chapter 24 of his
novel \textit{Blaquerna}) and \textit{De Arte Eleccionis}, all
featuring variants of a pairwise election system we now know as
Condorcet voting~\cite{hagele:llull, szp:b:numbers}.




In the eighteenth century, the Marquis de Condorcet developed one of
the first criteria for evaluating voting systems. Condorcet's
criterion is that, for a given election, if there exists a candidate
that beats all other candidates in pairwise contests, then that
candidate must be the winner of the election. It is a somewhat
surprising result that such a candidate will not always exist due the
possibility of cycles in the pairwise societal preferences.  Condorcet
proposed this criteria and showed that many popular voting systems do
not meet it. Those that do are called Condorcet methods,
or Condorcet-consistent voting systems.  Elections using Condorcet
methods can be viewed as a number of pairwise majority-rule
elections\footnote{A majority rule decides on the alternative that receives the
  majority of the votes. } or in the case of an election with two
candidates, exactly equivalent to a majority-rule election. Hence the
literature often refers to Condorcet methods as variants of the
majority rule~\cite{risse:cordorcetborda}.  

Condorcet methods have often been contrasted with Borda voting, which
takes complete preference orderings as the votes and gives points to
each candidate based on the number of candidates they are ranked above
in each vote.  For instance a vote denoting $a$ is preferred to $b$ is
preferred to $c$ would give two points to $a$, one to $b$, and none to
$c$.  There are fervent arguments between proponents of the two systems, debating
the importance of the Condorcet criterion~\cite{newenhizen:bordarespectscondorcet,saari:Bordabetter,risse:cordorcetborda}, in a rivalry dating back to the Marquis de
Concorcet's criticism of Borda voting when it was first
introduced~\cite{szp:b:numbers}.

One persistent concern in elections is that some of the participants
may be able to vote strategically thus unfairly gaining an advantage over
honest voters.  
Pliny the Younger's attempts to manipulate the Roman senate circa 105
CE is possibly the earliest recorded instance of strategic behavior in
elections~\cite{szp:b:numbers}.  The senate, presiding over a murder trial, were
divided into three blocs: those who favored acquittal, banishment, or
death for the accused. The senators were more or less evenly
distributed among the three positions, with the acquittal bloc (headed
by Pliny) being slightly larger than the other two. The
normal method of voting was similar to the current justice system in
most countries: The senate would first vote on the guilt of the
accused, followed by a vote on the punishment (banishment or death).
Considering how the blocs were aligned, the probable outcome of the
first election would be a decision of guilty, followed by
banishment. To give his faction an edge, Pliny proposed the senators
vote for acquittal, banishment or death in a single ternary-choice
election. However, his strategy backfired. The death penalty faction,
fearing an acquittal, voted for banishment.

Pliny's story holds more than just strategy and counter-strategy:
Pliny convinced the senate of the fairness of a single ternary-choice
election by stating it aligned naturally with the principle of voting
\emph{qua sentitits}, or according to your true preferences.
His attempt proved unsuccessful but serves as an excellent example of
the problem of getting people to vote honestly. While this problem was
recognized by voting theoreticians through history, it was either
dismissed or attempts to solve it were limited at best.  For instance,
Llull documents that voters were required to give an oath to vote sincerely, and
Jean-Charles de Borda famously dismissed criticism of his system's
vulnerability to manipulation by saying ``My scheme is only intended for
honest men''~\cite{bla:b:polsci:committees-elections}.



Later work drew from game theory to more formally model voter strategy
and to analyze its possibilities, especially in the work of Allan
Gibbard~\cite{gib:j:polsci:manipulation} and Mark
Satterthwaite~\cite{sat:j:polsci:manipulation}.  We will first explore
the work of Kenneth Arrow, who initiated the modern study of voting
theory and introduced the election model that has become the standard.

\subsection{Arrow's Impossibility Theorem}
Arrow's seminal work in modern social choice theory~\cite{arrow,
  arr:b:polsci:social-choice} originated in an attempt to formalize an
aggregate function for social opinion. Aggregate mechanisms existed
previously in welfare economics: Called welfare functions, they
attempted to measure societal welfare for a number of alternative
possibilities by aggregating the utility or
welfare of individuals (measuring, for instance, the impact of  a change in
fiscal policy or tax rates).  These mechanisms all shared the assumption
that as subjective a concept as individual utility could be compared
or quantified. A breakthrough came with the Bergson-Samuelson
social welfare function~\cite{bergsonswf} which inspired Arrow's own
aggregate mechanism, also called a social welfare
function\footnote{Differences between the two functions are elaborated on throughout Arrow's paper~\cite{arrow}.}. Like the
Bergson-Samuelson model, Arrow broke from previous economic models by
considering an individual's vote to be their ranked preferences rather than a collection of 
numerical utilities over the alternatives.  Thus the output of the
social welfare function is a ranked ordering of the alternatives as
well.  Arrow argues that this is a more appropriate model for aggregate functions
due to the difficulty of interpersonal comparisons of utility.



Arrow's key result, known as Arrow's impossibility theorem, shows that
no social welfare function can satisfy all of a set of five reasonable
criteria whenever there are more than two alternatives.  By reasonable
criteria we mean these conditions ``\ldots accord with common sense
and with our intuition about fairness and the democratic
process"~\cite{szp:b:numbers}.  In formalizing his aggregate
mechanism, Arrow laid down two postulates that directed the
construction of individual preferences, and outlined the
aforementioned five characteristics.

The first postulate states that for any pair of alternatives $a,b$,
every individual will always have some opinion between them:
individuals can be indifferent between them (generally represented as
$aIb$ or $bIa$, denoting indifference between $a$ and $b$), or prefer
one alternative to the other (generally represented as $aPb$ in the
case that $a$ is preferred to $b$).  The second postulate is that an
individual's preferences must be transitive, thus disallowing cycles
in individual preference orderings. Note that this requirement is not
universally held in voting theory, and intransitive preferences are
sometimes considered reasonable when voters use different criteria to
decide between different pairs of
alternatives~\cite{hug:j:rational}. Consider the example of an
individual Jeff who has to rent a car, and is willing to pay a little
more for additional space. Between a compact and a midsize car, Jeff
always chooses the midsize, since he has to pay just a little more for
additional comfort. Similarly, between a midsize and a fullsize car,
Jeff prefers a fullsize car. But between a fullsize car and a compact,
Jeff finds the price difference to be too great, and chooses the
compact car instead.


Arrow defines five reasonable criteria for social welfare
functions: unrestricted domain, monotonicity, nonimposition,
independence of irrelevant alternatives and nondictatorship.


\paragraph{Unrestricted Domain}
By the two aforementioned postulates, an individual's preferences are
represented as an ordering complete over the set of alternatives. Any
restriction on which sets of orderings are permitted as input to the
function violates the democratic nature of the mechanism and would
fail to satisfy this criterion. Unrestricted domain would be violated
in the example of elections held in a despotic state where only votes
with the current ruler ranked first are allowed.

%

\paragraph{Nonimposition}
The second criterion states that the function should not allow inclusion or preclusion of
outcomes irrespective of the preferences of the electorate. This criterion implies the social outcome
should depend entirely on the set of individual preference orderings. An example of imposition could be an election in a theocracy where only candidates from the state religion are permitted to be elected.

\paragraph{Monotonicity}
The third property, monotonicity, states that an individual cannot harm an alternative's position by ranking it higher. In other words, if an aggregate preference ordering holds that alternative Ted is preferred to
alternative Jeff, then an individual cannot harm Ted's aggregate
position by ranking him above Jeff in his or her ordering, all other
individual orderings being constant. This property implies that that
the aggregate decision must be responsive to and representative of the
individual's preferences.

In the later version of his work, Arrow replaced monotonicity and
nonimposition with the combined property of the Pareto criterion, or
unanimity~\cite{arr:b:polsci:social-choice}. The Pareto criterion is
slightly stronger than monotonicity since it incorporates
nonimposition. It states that for any two alternatives $a$ and $ b$,
if an individual preference ordering prefers $a$ to $b$ with all other
individual preferences indifferent between these two alternatives,
then the social outcome also prefers $a$ to $b$.

\paragraph{Independence of Irrelevant Alternatives}
The fourth property, independence of irrelevant alternatives (IIA),
implies that individual preferences for any pair of alternatives
should not be influenced by other alternatives. A famous anecdote
attributed to Sidney Morgenbesser~\cite{poundstonegaming} illustrates
IIA: Morgenbesser, ordering dessert in a restaurant, was informed by
the waitress that the dessert choices were apple pie and blueberry
pie. Morgenbesser chose apple pie. A few minutes later, the waitress
returned and informed him that cherry pie was also available. ``In that
case," said Morgenbesser to the utter confusion of the poor waitress,
``I'll have blueberry."



IIA and the possibility of strategic behavior in a voting system are mutually exclusive:
the presence of one in a voting model indicates the absence
of the other. While any honest preference relation between a pair of
alternatives would not be influenced by extraneous alternatives,
strategic behavior often requires them. Consider a Borda election
between two alternatives Jeff and Mike where a certain voter prefers
Mike, the stronger candidate, to Jeff. A new candidate, Ted, is
introduced that the voter prefers to all others. The voter,
then, instead of his true preference ordering \mbox{Ted $>$ Mike $>$ Jeff} (where \mbox{Ted $>$ Mike} implies that Ted is preferred to Mike), might misrepresent his preferences as \mbox{Ted $>$ Jeff $>$ Mike}, in order to weaken Ted's strongest opponent, Mike. In other words, the introduction of Ted results in the voter switching his preferences for Jeff and Mike.

\paragraph{Nondictatorship} 
The fifth property states the function should not permit an individual who
is a dictator, i.e., for a given profile of individuals, the function
cannot reflect any one individual's preferences, irrespective of the
preferences of all others in the profile.

Arrow proved that any thusly defined acceptable social welfare
function, meeting all these criteria, cannot decisively aggregate the
preferences of a profile of individuals if there are more than two
alternatives. This result essentially meant that any
social welfare function violated the most basic thresholds for
acceptability, thus any such function would have to compromise on meeting at
least one of these five criteria. Arrow, in discussing
this problem opined that compromising on an unrestricted
domain was the only reasonable alternative~\cite{arrow}.

One of the more notable approaches to this problem actually preceded
Arrow's work: In 1948 Scottish economist Duncan Black wrote
about an intriguing property of societal preferences called
single-peakedness~\cite{blacksinglepeak, bla:b:polsci:committees-elections}. Consider
plotting a preference ordering with the horizontal axis representing a
linear ordering of alternatives and the vertical axis representing
their rank in the ordering. If the resulting curve (drawn from joining
all alternative-representing points together) has a single peak
(defined to be a point flanked by either lower-ranking points on both
sides, or only on one side if the peak starts or ends the curve)
then we can state that the preference ordering is single-peaked with respect to
the linear ordering on the horizontal axis. For a given profile
of preferences, if there exists at least one linear ordering such that
all votes are single-peaked with reference to this linear ordering,
then we pronounce the profile to be single-peaked, or admitting the
property of single-peakedness.

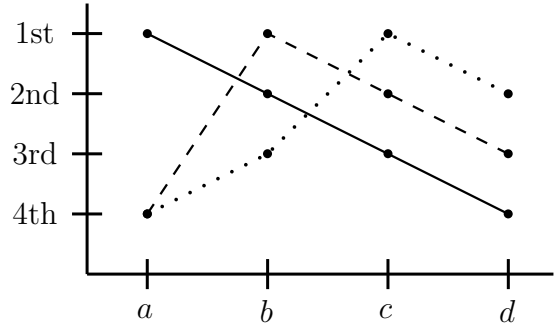
\begin{figure}[ht]
\center
\scalebox{1} 
{
\begin{pspicture}(0,-2.1425)(7.014375,2.1225)
\psline[linewidth=0.04cm](0.794375,2.1025)(0.794375,-1.4975)
\psline[linewidth=0.04cm](0.794375,-1.4975)(6.994375,-1.4975)
\psline[linewidth=0.04cm](1.594375,-1.6975)(1.594375,-1.2975)
\psline[linewidth=0.04cm](0.994375,-0.6975)(0.594375,-0.6975)
\psline[linewidth=0.04cm](3.194375,-1.2975)(3.194375,-1.6975)
\psline[linewidth=0.04cm](4.794375,-1.2975)(4.794375,-1.6975)
\psline[linewidth=0.04cm](0.994375,0.1025)(0.594375,0.1025)
\psline[linewidth=0.04cm](0.994375,0.9025)(0.594375,0.9025)
\rput(1.5675,-1.9875){$a$}
\rput(3.1785936,-1.9875){$b$}
\rput(4.762344,-1.9875){$c$}
\psline[linewidth=0.04cm](6.394375,-1.2975)(6.394375,-1.6975)
\rput(6.3789062,-1.9875){$d$}
\psline[linewidth=0.04cm](0.994375,1.7025)(0.594375,1.7025)
\psdots[dotsize=0.12](1.594375,1.7025)
\psdots[dotsize=0.12](3.194375,0.9025)
\psdots[dotsize=0.12](4.794375,0.1025)
\psdots[dotsize=0.12](6.394375,-0.6975)
\psline[linewidth=0.03](1.594375,1.7025)(3.194375,0.9025)(4.794375,0.1025)(6.394375,-0.6975)
\psdots[dotsize=0.12](3.194375,1.7025)
\psdots[dotsize=0.12](4.794375,0.9025)
\psdots[dotsize=0.12](6.394375,0.1025)
\psdots[dotsize=0.12](1.594375,-0.6975)
\psline[linewidth=0.03,linestyle=dashed,dash=0.16cm
  0.16cm](1.594375,-0.6975)(3.194375,1.7025)(4.794375,0.9025)(6.394375,0.1025)
\psdots[dotsize=0.12](6.394375,0.9025)
\psdots[dotsize=0.12](4.794375,1.7025)
\psdots[dotsize=0.12](3.194375,0.1025)
\psdots[dotsize=0.12](1.594375,-0.6975)
\psline[linewidth=0.05,linestyle=dotted,dotsep=0.16cm](6.394375,0.9025)(4.794375,1.7025)(3.194375,0.1025)(1.594375,-0.6975)
\rput(0.1,-0.6875){4th}
\rput(0.1,0.1125){3rd}
\rput(0.1,0.9125){2nd}
\rput(0.1,1.7125){1st}
\end{pspicture} 
}
\caption{A preference profile that is single peaked for the ordering
  $a b c d$.  The votes are $a > b > c > d$ (solid), $b > c > d > a$
  (dashed), and $c > d > b > a$ (dotted).}
\end{figure}

Black showed that aggregate functions admitting single-peaked
preference profiles (with respect to some linear ordering) meet all of
Arrow's criteria except for unrestricted domain. We lose this
criterion since a linear ordering that induces single-peakedness does not
exist for every preference profile\footnote{The existence of a
  single-peaked profile would (partially) depend on the set of
  least-preferred \mbox{alternatives} across all given orderings. For
  example, in the case of exactly three alternatives $a,b,c$, any
  profile of two or more preference orderings having a total of two
  alternatives $a,b$  ranked last, linear orderings
  $acb$ and $bca$ exist relative to which the profile is
  single-peaked. For methods to determine single-peakedness, we refer
  the reader to the work of Ballester and
  Haeringer~\cite{ballester:single-peaked} and Escoffier et
  al.~\cite{esc-lan-ozt:c:single-peaked-consistency}.}.
Therefore, if an aggregate function is to admit only single-peaked
input, it would  have to exclude certain preference orderings,
restricting the domain.

Another approach by Amartya Sen~\cite{senQuasi} showed the existence of
aggregate mechanisms that have all of Arrow's criteria but implement
a relaxation of transitivity called quasi-transitivity.  This permits
the existence of certain preferences that violate standard
transitivity---for example, for alternatives $a, b, c$ the
following preference profile is acceptable: $aIb$, $bIc$, $aPc$. Sen also
discusses an aggregate function where the input, instead of individual
preference orderings, is individual utility functions.

Brams and Fishburn showed that approval voting similarly has very
desirable properties when we restrict the preference
domain~\cite{nie:j:approval}.  In the case that voters have
dichotomous preferences (that is, they can divide the candidates into
two groups: one they they approve of and one they do not), approval
voting has many positive properties, including immunity to strategic
voting and the Condorcet criterion~\cite{bra-fis:j:approval}.  In the
general case with unrestricted preferences, the system no longer has
these properties~\cite{nie:j:approval}.

The impact of Arrow's work can be summed up in a quote attributed to
Paul Samuelson~\cite{poundstonegaming}, ``What Kenneth Arrow proved
once and for all is that there cannot possibly be \ldots an ideal voting
scheme.'' The Gibbard-Satterthwaite theorem held even more dismal news:
In addition to being less-than-ideal, all voting schemes are also
vulnerable to manipulation, unless they admit dictators.

\subsection{Gibbard-Satterthwaite Theorem}
In the 1970s Alan Gibbard~\cite{gib:j:polsci:manipulation} and Mark
Satterthwaite~\cite{sat:j:polsci:manipulation} independently extended Arrow's theorem for voting systems in a model that incorporated strategic misrepresentation of preferences. They proved that no \emph{strategy-proof} voting system existed for elections with three or more alternatives, unless the voting system allowed dictators. A strategy-proof voting system is one where no manipulating strategy can exist in elections using this voting system. While this result revolutionized voting theory, it had been speculated previously: In 1960, William Vickrey, when discussing individuals strategically misrepresenting their preferences in Arrow's model~\cite{vickrey:gsconjecture}, stated that ``it is clear that social welfare functions that satisfy the non-perversity [monotonicity] and the independence [IIA] postulates and are limited to rankings as arguments, are also immune to strategy." In addition, the Dummett-Farquharson conjecture of 1961~\cite{dummettfarquharsonstability} parallels the Gibbard-Satterthwaite theorem. 

The Gibbard-Satterthwaite theorem transformed voting theory for two reasons: one was the aforementioned result that nondictatorial voting rules were susceptible to strategic voting (manipulation) in cases with three or more outcomes,  and the second was the adoption of the arcane science of game theory from Neumann
and Morgenstern's \textit{Theory of Games and Economic Behavior}~\cite{neumanngames}, published just four years prior to
Arrow's work. While the influence of game theory was implied (and acknowledged) in Arrow's social welfare
function~\cite{wilsonarrowgt}, the Gibbard-Satterthwaite theorem proofs were more explicit in their treatment of voting
functions as game-theoretic mechanisms. This provided the field of voting theory with a set of tools to examine a whole range of scenarios---for example, the motivations for the electorate to vote dishonestly, or their reaction to changes in the structure of the voting rule, or attempts to form coalitions to further their individual utility. 

In Gibbard's work, a \emph{voting scheme} or \emph{social choice
  function} is built upon a construct called a \emph{game
  form}~\cite{gib:j:polsci:manipulation}.  A game form is similar to a
game construct in game theory~\cite{neumanngames} applied to a voting
model: players are voters, and any player's strategies are the set of
all possible orderings of preferences (unless restricted for specific
scenarios), over a set of alternatives or candidates. Game forms,
unlike games, do not have functions assigning utilities for each
player for a given action (or chosen strategy) or for a scenario of
chosen strategies of all the players. Instead the social choice
function has the concept of an honest or sincere strategy: Among the
set of strategies for each player is a specially marked strategy
denoting the honest representation of preferences for that
player. This can be used to compare outcomes (which may be a
preference ordering or a single alternative) for a voter's different
strategies, to see if one is less or more preferable to another in the
honest ranking of alternatives.

This social choice function will be immune to manipulation only if each
voter has a \emph{dominant strategy}, a strategy that will be at least
as good as any other for that voter no matter what any other voter
does.  Otherwise, if a voter does not have a dominant strategy, then
they might possibly be motivated to change their vote from their true
preferences in order to obtain a better outcome.  A social choice function where
each voter has a dominant strategy (and hence is immune to manipulation) is said to be \emph{straightforward}.

The proof of the Gibbard-Satterthwaite theorem relies on the Vickrey
conjecture: a voting scheme (defined with the property of unrestricted
domain) is strategy-proof if and only if it has the properties of IIA and unanimity. Since Arrow showed that no aggregate function with more than two outcomes can satisfy all of his model's criteria, such a voting system then is necessarily a dictatorship.
In other words, we have that any voting system with at least three outcomes will either be a
dictatorship or it will be manipulable. A similar result, the Duggan-Schwartz theorem~\cite{dug-sch:j:polsci:gibbard}, exists for  voting systems that elect multiple candidates.

The Gibbard-Satterthwaite theorem presents a problem and accepts a
solution similar to Arrow's theorem. By Black's results, voting
schemes that permit only single-peaked preferences restrict the domain
of the function, but can have all other Arrow criteria, including
unanimity and IIA, which together imply strategy-proofness. Thus,
relaxing the condition of unrestricted domain is necessary for voting
schemes that resist manipulation. Another example of this is Gibbard's
development of probabilistic
mechanisms~\cite{gibbard:relaxingunrestricteddomain:chance,
  gibbard:relaxingunrestricteddomain:lotteries}.  Procaccia took a
similar approach by designing strategy-proof probabilistic voting
systems that are similar to standard deterministic voting systems~\cite{DBLP:conf/aaai/Procaccia10}.

\subsection{Computational Difficulty of Manipulation}






Another solution to the problem of inherent manipulability in voting
was proposed by Bartholdi along with Tovey, Trick, and Orlin in a
series of papers that started the field of
computational social choice~\cite{bar-tov-tri:j:manipulating,bar-tov-tri:j:who-won,
  bar-oli:j:polsci:strategic-voting}.
Their approach was to select voting schemes
where manipulation is computationally difficult to carry out,
i.e. where the manipulation problem is NP-hard. Our definition of the
manipulation problem is that of constructive coalitional manipulation:
i.e., does there exist a set of votes for the manipulating coalition
that causes their preferred candidate to win the election? This
subsumes the case of a single manipulator and contrasts with
destructive manipulation, which is concerned with preventing a certain
candidate from winning.

Bartholdi et al.'s initial
work  also highlighted the problems of
selecting systems with complexity. Their impracticality
theorem~\cite{bar-tov-tri:j:who-won} is another instance of systems
meeting seemingly reasonable criteria thereby inducing  undesirable
properties.

%
The theorem states that any \emph{fair} voting system requires excessive computation to determine the
winner, making it impractical---a highly disturbing result.
This theorem followed work by Kemeny~\cite{kem:j:no-numbers},  Young and Levenglick~\cite{lev-you:j:condorcet} and Gardenfors~\cite{gardenfors:manipscfprecedesimpracticalitytheorem}. 

According to Bartholdi et al., a voting system is \emph{fair} if it
meets the Condorcet criterion, the condition of \emph{neutrality}
(symmetry in its treatment of candidates, implied by
IIA~\cite{gal-pat-pen:j:arrow}), and the condition of
\emph{consistency} (if disjoint subsets of the voters voting
separately arrive at the same preference ordering, then voting
together always produces this same preference ordering as well). Their
theorem states that computation of the winner in any such fair voting
system is NP-hard.

The previously mentioned property of consistency (also called
convexity~\cite{woodall:convexityproperty} and
separability~\cite{smith:separability}) has been proven to be present
in ranked voting systems (those in which a vote is a ranking or
ordering of preferences) only if they happen to be scoring protocols
as well (where alternatives receive a certain number of points
depending on their position in the
ordering)~\cite{young:consistencypositionalpointscoring}. Scoring
protocols are often incompatible with the Condorcet criterion (refer
the aforementioned debate on Condorcet versus Borda) thus
unsurprisingly so far we know of only one voting system, Kemeny
scoring~\cite{kem:j:no-numbers}, that meets all three
conditions~\cite{lev-you:j:condorcet}.
Kemeny scoring then, is NP-hard, but additionally it has been shown to
be complete for parallel access to
  NP~\cite{hem-spa-vog:j:kemeny}.


However, this does not imply other voting systems are immune to having
an intractable winner problem: systems such as Dodgson meet only two
of the three fairness conditions---that of the Condorcet criterion
and neutrality but not consistency---however computation of the winner
in Dodgson is known to be not only
NP-hard~\cite{bar-tov-tri:j:who-won} but complete for parallel access
to NP~\cite{hem-hem-rot:j:dodgson}, similar to the result for Kemeny scoring.
 
Research focusing on tractable voting systems\footnote{A voting system
  is tractable if calculating the winner takes at most polynomial time.} was more
promising: While Bartholdi et al. gave us a greedy algorithm that
finds a manipulating vote for several tractable voting systems in
polynomial time~\cite{bar-tov-tri:j:manipulating}, two voting systems---second-order Copeland and single
transferable vote~\cite{bar-tov-tri:j:manipulating,
  bar-oli:j:polsci:strategic-voting}---proved to be resistant and were
shown to have a manipulation problem that is NP-hard. Research in this
field remained dormant for the next fifteen years until a revival
starting in 2006 brought about results for most common tractable
voting systems.

In this paper we are concerned with a restricted version of the manipulation
problem.  We survey tractable voting systems that resist manipulation in the
unweighted coalitional manipulation (UCM) model with only a constant
number of manipulators.  
This limited case subsumes hardness results in the weighted coalitional manipulation (WCM) model or with
variably-sized coalitions, thus making a case for UCM being a stronger class of manipulation. 
We include both the initial work of Bartholdi, Tovey, and Trick~\cite{bar-tov-tri:j:manipulating} and
Bartholdi and Orlin~\cite{bar-oli:j:polsci:strategic-voting} that
achieved the first results in this area and the recent resurgence of
interest in this problem that has resulted in a number of new
outcomes.

\subsection{Terms Defined}

An \emph{election} is defined to be an instance of a voting system,
comprising a voting rule $vr$, a set of candidates $C$ and a set of
votes $V$.  A \emph{voting rule} is a function that takes as input a
set of votes and a set of candidates and outputs a set of winners.
Unless explicitly stated otherwise, references to $n$ and $m$ are
defined as follows: $n = ||V||$ and $m = ||C||$. A \emph{vote} is
defined to be a linear ordering over the set of candidates.  The
\emph{advantage} of a candidate $c_{i}$ over $c_{j}$ (hereafter
referred to as $adv(c_{i},c_{j})$) is the number of votes that rank
$c_{i}$ ahead of $c_{j}$, with values ranging from $0$ to $n$. The
\emph{net advantage} of a candidate $c_{i}$ over $c_{j}$ is
$adv(c_{i},c_{j}) - adv(c_{j},c_{i})$, with values ranging from $-n$
to $n$.  A $netadv$ score between two candidates $c_{i}$ and $c_{j}$
is represented as $netadv(c_{i},c_{j})$.
By definition we can see that
$netadv(c_{i},c_{j}) = - netadv(c_{j},c_{i})$, thus one $netadv$
score can represent both directions.

UCM (unweighted coalitional manipulation) is defined to be a decision
problem as follows.
\begin{description}
\item[Given] An election, namely, a voting rule $vr$, a set of voters
  $V$ such that $V = V_{NM} \cup V_{M}$, where $V_{M}$ is the subset
  of voters that form the manipulating coalition and $V_{NM}$ is all
  other voters, and a set of candidates $C$ containing a
  distinguished candidate $c$.
\item[Question] 
   Does there exist a set of votes for $V_{M}$ such
that $vr$ over the complete set of votes yields $c$ as the winner?
\end{description}

 We use the format UCM$_{2Cope}$ to refer to the UCM problem in
 second-order Copeland and likewise for other election systems.

\section{UCM in Single Transferable Vote}

Single transferable vote (henceforward STV) is a voting system
with a long history.  As esteemed a figure as John Stuart Mill said it
was ``among the greatest improvements yet made in the theory and
practice of government.'' It determines the winners with a simple
multiround procedure that redistributes votes placed for less
popular candidates.  Also, unlike many of the
esoteric voting systems studied in voting theory, STV has a
history of being used for real-world political elections, in the
United States and around the world.  

The STV vote tallying procedure is as follows.  Give a point to each
candidate for each first-place vote it receives.  If any
candidate is the majority winner (i.e.\ with more than half the total
points), that candidate will be the only winner of the election.  If no
majority winner exists, then select the candidates with the fewest
number of points, remove them from consideration, and for the voters
who currently give their support to these candidates, reallocate
their support by giving their points to the next-highest ranked candidate
on their ballots still under consideration.  Repeat this procedure
until a winner is chosen or all candidates are removed.  If the latter
occurs, then all of candidates that were removed in the last round
will be the winners.

The complex, shifting behavior of STV with multiple candidates is what
gives it the resistance to manipulation we discuss here, but it also
leads to STV failing to possess some very desirable voting system
characteristics.  Notably it does not possess the property of
monotonicity: It is possible for a voter to increase his or her ranking of
a candidate and for that candidate to subsequently do worse in the
election.  This was enough for Doron and Kronick~\cite{perverse:stv}
to refer to it as a ``perverse social choice function,'' and it
certainly is a flaw of concern.

\subsection{STV is Resistant to Manipulation}


We show resistance to manipulation through a conventional, if
difficult, reduction based on Bartholdi and Orlin's work~\cite{bar-oli:j:polsci:strategic-voting}.
Their work was actually directed towards showing that
the EFFECTIVE PREFERENCE problem is NP-complete. EFFECTIVE PREFERENCE
is the problem of finding if a single voter can cause the preferred
candidate to win an election.  This is effectively the same as UCM
with a single manipulator, and we will prove that this problem is
NP-hard for STV with a proof based on the aforementioned
work~\cite{bar-oli:j:polsci:strategic-voting}.  The proof is
structured as a reduction from the exact cover by three-sets problem.

\noindent
\textbf{Exact Cover by Three-Sets (X3C)}

\label{sec:x3cdef}
\begin{description}
\item[Given] A set $D = \{d_1, \ldots, d_{3k}\}$ and a family ${\cal
  S} = \{S_1, \dots, S_n\}$ of sets of size three of elements from $D$.
\item[Question] Is it possible to select $k$ sets from ${\cal S}$
  such that their union is exactly $D$?
\end{description}

In other words, the goal of the problem is to find if there is an
appropriately-sized set of subsets which covers each of the elements
in $D$.  Since each $S_i$ has exactly three elements and $k$ such sets
from ${\cal S}$
must be chosen, the set of subsets must be an exact cover with no
repeated elements in all subsets chosen.

\begin{proof}

We will describe a reduction from an instance of X3C to a instance
of the unweighted manipulation problem for STV\@.  Note that since the reduction
will only require a single manipulator, this shows that UCM$_{STV}$ is NP-hard even for only a single
manipulator.

Given an instance of X3C $(D, {\cal S})$ we construct the
election as follows.  Let the following comprise the candidate set
$C$:

\begin{itemize}
\item The possible winners $c$ and $w$;
\item The set of ``first losers'' $a_1, \dots, a_n$ and
  $\overline{a}_1, \dots \overline{a}_n$, one of each corresponding to
  each subset $S_i$;
\item The ``second line'' $b_1, \dots, b_n$ and
  $\overline{b}_1, \dots, \overline{b}_n$, one of each corresponding to
  each subset $S_i$;
\item The $w$-bloc $d_0, \dots, d_{3k}$, each of whose voters will just
  prefer them to $w$;
\item The ``garbage collector'' candidates $g_1, \dots,  g_n$.
\end{itemize}

We will now describe the set of voters.  Where we use ellipses, the
remainder of a vote is arbitrary for our purposes and will not effect
the result of the election.

\begin{itemize}
\item $12n$ voters with preferences $(c, \dots)$;

\item $12n-1$ voters with preferences $(w,c,\dots)$;

\item $10n + 2k$ voters with preferences $(d_0, w, c, \dots)$;

\item For each $i \in \{1, \dots, 3k\}$, $12n - 2$ voters with
  preferences $(d_i, w, c, \dots)$;

\item For each $i \in \{1, \dots, n\}$, $12n$ voters with preferences
  $(g_i, w, c, \dots)$;

\item For each $i \in \{1, \dots, n\}$, $6n + 4i - 5$ voters with
  preferences $(b_i, \overline{b}_i, w,c, \dots)$ and for the three $j$
  such that $d_j \in S_i$, 2 voters with preferences $(b_i,
  d_j, w, c, \dots)$;

\item For each $i \in \{1, \dots, n\}$, $6n + 4i - 1$ voters with
  preferences $(\overline{b}_i, b_i, w, c, \dots)$ and 2 voters with
  preferences $(\overline{b}_i, d_0, w, c, \dots)$;

\item For each $i \in \{1, \dots, n\}$, $6n + 4i - 3$ voters with
  preferences $(a_i, g_i, w, c, \dots)$, 1 voter with preferences $(a_i, b_i,
  w, c, \dots)$, and 2 voters with preferences $(a_i, \overline{a}_i, w, c, \dots)$;

\item For each $i \in \{1, \dots, n\}$, $6n + 4i - 3$ voters with
  preferences $(\overline{a}_i, g_i, w, c, \dots)$, 1 voter with preferences
  $(\overline{a}_i, \overline{b}_i, w, c, \dots)$, and 2 voters with
  preferences $(\overline{a}_i, a_i, w, c, \dots)$.
\end{itemize}

This reduction works by requiring the elimination order of a subset of
the candidates to correspond to an exact cover over $B$ in order for
$c$ to win the election.  Namely, $c$ will win the election if and
only if \mbox{$I =
\{i ~|~ b_i $ is one of the first $3n$ candidates to be eliminated$\}$} is an exact cover.  Furthermore, there is a preference order for a
single manipulator that will force $I$ to be an exact cover if one
exists.  We will now consider the relevant properties of the election
and show that this is the case.

Since this election is a single-winner STV election with more than two
candidates, the scoring process will proceed for a number of rounds
and a number of candidates will be eliminated as the rounds progress.
The first $3n$ candidates to be eliminated will be $a_1, \dots, a_n$,
$\overline{a}_1, \dots, \overline{a}_n$, and exactly one of $b_i$ or
$\overline{b}_i$ for every $i \in \{1, \dots, n\}$.

Candidate $c$ initially has $12n$ votes, while $c$'s primary rival $w$
has $12n-1$ votes.  Every voter that does not have $c$ as their first
choice ranks $c$ directly below $w$, and so $c$ can only gain more
votes if $w$ is eliminated.  In order to do so, the manipulator must
ensure that $w$ does not gain additional votes before it is
eliminated, as otherwise $w$ would have gained votes against $c$ and
$c$ would have been eliminated first.  The manipulator must
consequently make sure that no candidate is eliminated such that any
voter most prefers that candidate and prefers $w$ second-most.  This
is the case with every voter that prefers one of the $d_j$ candidates,
so $w$ will gain a large number of votes if one of them is eliminated.
Therefore if any $d_j$ candidate is eliminated before $w$, $c$ cannot
possibly win.  For every $b_i$ candidate that is eliminated,
$\overline{b}_i$ gains a large number of votes, pushing it higher than
$12n$ in score and preventing it from being eliminated early. Also,
every $d_j \in S_i$ gains two votes, pushing them high
enough to prevent their early elimination.  Thus eliminating $b_i$
protects the $d_j$ candidates associated with the set $S_i$.
Conversely, for every $\overline{b}_i$ that is eliminated, $d_0$ gains
two votes and $b_i$ gains a large number of votes, preventing $b_i$
from being eliminated early.  Thus for every $i$ only one of $b_i$ or
$\overline{b}_i$ can be eliminated before $c$ or $w$.

The $a$ candidates are the other candidates that can be eliminated
early.  For every $a_i$ that is eliminated, $\overline{a}_i$ gains two
votes, $b_i$ gains one vote, and $g_i$ gains the rest of $a_i$'s votes.
The effect of this is that now $\overline{a}_i$ has been promoted
above $b_i$ and $\overline{b}_i$ in the overall ranking, and since
$b_i$ also gained a point over $\overline{b}_i$, $\overline{b}_i$ will
be the next to be eliminated instead of $b_i$.  Thus by controlling
which of $a_i$ or $\overline{a}_i$ is eliminated first, we control
which of $b_i$ or $\overline{b}_i$ is eliminated early.

Hence we can show that the candidate $c$ will win the election if and
only if $I = \{i
~|~ b_i $ is one of the first $3n$ candidates to be
  eliminated$\}$ is an exact cover.  We know that either $b_i$ or
$\overline{b}_i$ will be among the first $3n$ eliminated candidates.
If $b_i$ is the one eliminated, then every $d_j \in S_i$ will
gain two votes and will each have at least $12n$ votes total, protecting
them from early elimination.  If $I$ is an exact cover, this will be true
for every $d_j$ as each of them is covered by some selected $b_i$ and
so each of them will win over $w$.  Also, since $d_0$ gains two points
for every one of the $\overline{b}_i$ eliminated, $d_0$ will gain at
least $2(n-k)$ votes and will receive at least $12n$ votes overall,
pushing it over the score of $w$ as well.  Thus after the first $3n$
candidates have been eliminated, $w$ will have the least score with
$12n-1$, having gained no votes other than it's initial first-place
votes, and it will be the next candidate to be eliminated.  The
candidate $c$ will then gain a large number of votes from the
elimination of $w$ and will go on to win the election.  

If the set $I$ defined above does not correspond to an exact cover, $c$
cannot win the election.  If $I$ is not an exact cover, some candidate
$d_j$ will not gain the two points from a corresponding $b_i$ being
eliminated.  Thus $d_j$ will only have $12n-2$ votes after
the first $3n$ candidates are eliminated while the other remaining
candidates have at least $12n-1$, leaving $d_j$ as the next candidate to be
eliminated.  The candidate $w$ then gains points from $d_j$'s elimination,
preventing $c$ from gaining points against $w$ and winning the election.

If an exact cover exists, a single manipulator can construct it's
preference order as follows to ensure $c$ is a winner.  For an exact
cover $I$, if $i \in I$, let the $i$th candidate in the preference
order be $a_i$ and otherwise $\overline{a}_i$.  The rest of the
preference order is arbitrary.  This will result in the following
order of elimination of the first $3n$ candidates: For $i \in I$,
$\overline{a}_i, b_i, a_i$ will be eliminated in that order in
positions $3i-2$, $3i-1$, and $3i$.  For $i \notin I$, $a_i,
\overline{b}_i, \overline{a_i}$ will instead be the candidates to be
eliminated.  For $i \in I$, since this preference ranks $a_i$ over
$\overline{a}_i$, $a_i$ will gain one more vote and thus
$\overline{a}_i$ will be eliminated first.  This then gives two votes
to $a_i$ and one more to $\overline{b}_i$, making $b_i$ the next
least-preferred candidate.  When $b_i$ is eliminated, $\overline{b}_i$
gains a large number of votes, so $a_i$ is now the least-preferred
candidate and is eliminated next.  The rounds proceed similarly in the
case that $i \notin I$ and $\overline{a}_i$ is preferred instead.
Thus the set $\{i ~|~ b_i \text{ is one of the first } 3n \text{
  candidates to be eliminated} \}$ will correspond to an exact cover
and $c$ will win the election.

If no exact cover exists, no matter how a manipulator votes, at least
one of the $d_j$ candidates will not receive the protective boost from
the elimination of the corresponding $b$ candidate as previously
described.  This candidate will then be  eliminated early, leading to
$w$ being boosted past $c$ in votes and preventing $c$ from winning.

Thus even just setting the first $n$ rankings for the ballot of a
single manipulator for STV is NP-hard, and the system is resistant
even to
this very limited case of manipulation.
\end{proof}

\section{UCM in Borda Voting}

Borda voting is a classic voting system dating back at least to the
eighteenth century.  It was introduced by the French mathematician and
engineer Jean-Charles de Borda to remedy the failure of plurality in
reflecting the wishes of the electorate when used with more than two
candidates:  In plurality the candidate with the most votes is not
necessarily preferred to all other candidates.  Borda voting is very
similar to a system introduced in the 15th century by Cardinal
Nicolaus Cusanus~\cite{szp:b:numbers}.

It has a rich and varied history of real-world use: In some form it has
been used in political elections in Slovenia and the Micronesian
countries of Kiribati and Nauru, in the Eurovision contest, the
election of the board of directors of the X.Org foundation, and even
in sports, in the election of the Most Valuable Player award in Major
League Baseball. Borda is one of a class of systems known as scoring
protocols, where each vote awards points to each candidate depending
on their ranking in the vote. The winners are candidates with the
highest sum of points over all the votes.  In the case of Borda
voting, candidates receive linearly descending points for
progressively less favorable positions in the votes, with the top
position awarding $m-1$ points, the next position awarding $m-2$
points, and so on down to $0$ points for the lowest position.

It was long an open problem whether UCM$_{Borda}$ was hard in
general, though manipulation with a single manipulator has long been
known to be easy~\cite{bar-tov-tri:j:manipulating}. A greedy
algorithm that can find a set of successful manipulating votes in polynomial time that is at most one larger than the optimum
manipulative coalition size is known as well~\cite{pro-ros-zuc:j:borda}.  Recently,
Betzler et al.~\cite{bet-nie-woe-c-borda} and Davies et
al.~\cite{dav-kat-nar-wal:j:borda} proved Borda-manipulation to
be NP-hard for instances with two or more manipulators.

\subsection{Borda is Resistant to Manipulation }
Both Betzler et al.~\cite{bet-nie-woe-c-borda} and Davies et al.~\cite{dav-kat-nar-wal:j:borda}  prove their results by reduction from
the problem of 2-numerical matching with target sums, a known NP-hard
problem~\cite{yu-2nmts} that closely corresponds to the problem of
allocating points to the nonfavored candidates in the election.

\noindent
\textbf{2-Numerical Matching with Target Sums (2NMTS)}
\begin{description}
\item[Given] A sequence $a_1, \dots, a_k$ of positive integers with
  $\sum_{i=1}^{k} a_i = k(k+1)$ and $1 \leq a_i \leq 2k$.
\item[Question] Are there  permutations $\psi_1$ and $\psi_2$ of
  $1,\dots, k$ such that $\psi_1(i) + \psi_2(i) = a_i$ for $1 \leq i
  \leq k$?
\end{description}

\noindent
\textbf{Preliminaries} For manipulation to work, nonfavored candidates
must be ranked low enough in manipulative votes such that the number
of points they gain by said votes do not prevent the preferred
candidate from winning. To that end we define the \emph{gap} to be the
maximum number of points nonfavored candidates can gain by
all manipulative votes while still allowing the preferred candidate to
win. In any Borda instance with a favored candidate $c^*$, $m$ other
candidates, and $t$ manipulators, the gap $g_i$ for a candidate $c_i$
is $score(c^*) + t \cdot m - score(c_i)$. Here $score(c)$ refers to the
Borda score for the candidate $c$ over the nonmanipulative votes.  We
assume these gap values $g_1,\dots,g_m$ to be ordered in a
nondecreasing fashion.

\noindent
\textbf{Result 1} Recall that the awarded points for the last $j$
candidates in a vote will range from $j-1$ to $0$, hence the sum of their
points equals $j(j-1)/2$. Thus for any successful manipulation
instance, we must have that $\sum_{i=1}^{j} g_i \geq t \cdot j(j-1)/2$
for each $j \in \{1,\dots,m\}$. We define an instance to be
\emph{tight} if $\sum_{i=1}^{j} g_i = t \cdot j(j-1)/2$. Thus, in a
tight instance, for manipulation to be successful, the number of
points a nonfavored candidate gains from manipulative votes must be 
exactly equal to its gap value.

\begin{proof}
 Given any instance of 2NMTS we construct a UCM$_{Borda}$ instance
 $(C,V, p)$ as follows.  The candidate set $C$ consists of candidates
 $c_{1}, \dots, c_{k}$ and the preferred candidate $p$, and thus the
 range of Borda points is from $0$ to $k$.  The set of votes $V$
 consists of a manipulating coalition of size two and a set of
 nonmanipulating votes of size three. The instance is constructed such
 that the gap $g_{i}$ for any nonfavored candidate $c_i$ is $2k -
 a_{i}$. In this context, the Borda problem can be considered as
 follows: for every nonfavored candidate $c_{i}$, can we assign a
 position in each of the manipulating votes such that the points $c_i$
 gains from said votes is $\leq g_{i}$? This constructed instance of
 Borda manipulation will have a solution if and only if the 2NMTS
 instance has a solution.

\noindent
\textbf{Direction 1} Given a solution to 2NMTS, we can obtain a
solution for the Borda instance as follows: Preferred candidate $p$ is
placed in the first position in both manipulating votes. A solution to
2NMTS exists so we have two orderings $\psi_1(i), \psi_2(i)$ such that
$\psi_1(i) + \psi_2(i) = a_{i}$. For every candidate $c_{i}, (1 \leq i
\leq k)$ set its position to $\psi_1(i) + 1$ in the first
manipulative vote, and $\psi_2(i) + 1$ in the second manipulative
vote. The corresponding Borda points are obtained from subtracting
this position number from $||C|| = k+1$. Therefore the points $c_{i}$ has gained
 from both manipulative votes is $(k+1-(\psi_1(i) + 1)) + (k+1
- (\psi_2(i) + 1))$ which equals $2k - a_{i} = g_{i}$, permitting $p$
to win.


\noindent
\textbf{Direction 2} Given a solution to the Borda instance, we have a
solution for 2NMTS as follows: By construction,
$\displaystyle\sum\limits_{i=1}^k g_{i} =
\displaystyle\sum\limits_{i=1}^k (2k - a_{i})$.

 Since
$\displaystyle\sum\limits_{i=1}^k a_{i} = k(k+1)$,
$\displaystyle\sum\limits_{i=1}^k (2k - a_{i}) = k(k-1)$.

 Hence,
this Borda instance is tight for $j=k$ and $t = 2$ (from Result 1) and
consequently each nonfavored candidate $c_{i}$ gains exactly $g_{i}$
points from the manipulative votes. If $pos_{i}(1)$ and $pos_{i}(2)$ are
the positions for $c_{i}$ in the two manipulative votes, points gained from these positions total $(k+1
- pos_{i}(1)) + (k+1 - pos_{i}(2)) = g_{i} = 2k - a_{i}$, which yields
$pos_{i}(1) + pos_{i}(2) = a_{i} + 2$. Therefore setting $\psi_1(i) =
pos_{i}(1) -1$ and $\psi_2(i) = pos_{i}(2) -1$ gives us a solution for 2NMTS.

As mentioned earlier, we assume by construction $g_{i} = 2k - a_{i}$.
This requires constructing the set of nonmanipulative votes such that
the deficits for each candidate relative to the preferred candidate
map precisely to the target sums in the original problem. Executing
this involves complicated construction and the addition of a large
number of ``dummy'' candidates to pad out the remaining positions and
precisely set the required deficits for the primary candidate
set~\cite{bet-nie-woe-c-borda} (in Davies et al., such padding is done
with voters~\cite{dav-kat-nar-wal:j:borda}).  Thus the constructed
instance has a much larger candidate set than a voter set (though it
remains polynomially bounded), but it suffices to prove the desired
hardness result.
\end{proof}

\section{UCM in Copeland Elections}

Copeland voting is a voting system with a long history.  One version
of the system was discovered by the 13th century mystic Ramon Llull,
and then another variation was discovered by A.H. Copeland in the
1950s.  It is a Condorcet voting system.

Copeland voting is in fact a family of voting systems, parametrized on
how ties are handled.  In Copeland$^\alpha$, the score of a candidate
$c$ in an election $E$ is wins$_E(c)$ + $\alpha \cdot $ties$_E(c)$, where
wins$_E(c)$ denotes the number of pairwise victories and ties$_E(c)$,
the number of ties of the candidate $c$ in the election $E$.  Llull's
system is Copeland$^1$, while ``Copeland voting'' has been used to
describe Copeland$^{0.5}$ or to describe Copeland$^0$.  Different
parameter values subtly alter the behavior of the system and
complicate the task of analyzing its computational properties, as we
will explore.



\subsection{Copeland is Resistant to Manipulation for Most Values of
  {\boldmath $\alpha$}}
UCM$_{Cope^{\alpha}}$ is in P when there is only one
manipulator~\cite{bar-tov-tri:j:manipulating}, but it is known to be
resistant to manipulation even with two manipulators for $\alpha \in
[0,0.5) \cup (0.5, 1]$~\cite{fal-hem-sch:c:manip-copeland,
  fal-hem-sch:c:copeland-ties-matter}. The complexity of UCM$_{Cope^{\alpha}}$ for $\alpha = 0.5$ remains unknown as of date. Different proofs were required to show resistance to manipulation for different parameter ranges, as
the behavior of the system changes in subtle but significant ways with
different values of the parameter.  We will describe the proof that
was used to show hardness for Copeland$^\alpha$ for $\alpha \in
\{0,1\}$~\cite{fal-hem-sch:c:manip-copeland}.
Both  these cases were proved by Faliszewski et
al.~\cite{fal-hem-sch:c:manip-copeland} through similar reductions
from X3C.

\noindent
\textbf{Exact Cover by Three-Sets (X3C)}

\begin{description}
\item[Given] A set $D = \{d_1, \ldots, d_{3k}\}$ and a family ${\cal
  S} = \{S_1, \dots, S_n\}$ of sets of size three of elements from $D$.
\item[Question] Is it possible to select $k$ sets from ${\cal S}$
  such that their union is exactly $D$?
\end{description}


\begin{proof}

The proof constructs a UCM$_{Cope^\alpha}$ instance from an X3C
instance $(D , {\cal S})$ using graph representations to equate both
problem instances. A election for the reduction can be constructed
without an explicit collection of votes as such a collection can be
elicited from the set of $netadv$ or $adv$ scores for all pairs of
candidates\footnote{Refer Appendix A.}. The graph representation of
the election can be constructed from the set of $netadv$ or $adv$
scores as well\footnote{Refer Appendix B.}. Both these constructions
are polynomially bounded in the size of the set of candidates and
the value of the $netadv$ function. Thus, using these techniques
we can construct the election given a partial set of significant
candidates, the $netadv$ scores for all pairs of candidates, and the
lead in Copeland score that the distinguished candidate has over each specified candidate.


\begin{figure}[h]
\centering

\scalebox{1} 
{
\begin{pspicture}(0,-2.11)(8.0078125,2.11)
\psdots[dotsize=0.12](1.5884376,1.3)
\psdots[dotsize=0.12](1.5884376,0.5)
\psdots[dotsize=0.12](1.5884376,-0.3)
\psdots[dotsize=0.12](1.5884376,-1.1)
\psdots[dotsize=0.12](3.1884375,0.1)
\psdots[dotsize=0.12](4.7884374,1.3)
\psdots[dotsize=0.12](4.7884374,0.1)
\psdots[dotsize=0.12](6.3884373,1.3)
\psdots[dotsize=0.12](6.3884373,0.1)
\psdots[dotsize=0.12](4.7884374,-1.1)
\psdots[dotsize=0.12](6.3884373,-1.1)
\psline[linewidth=0.04cm,arrowsize=0.1
  2.0,arrowlength=2,arrowinset=0.4]{->}(1.5884376,1.3)(3.1884375,0.1)
\psline[linewidth=0.04cm,arrowsize=0.1
  2.0,arrowlength=2,arrowinset=0.4]{->}(1.5884376,0.5)(3.1884375,0.1)
\psline[linewidth=0.04cm,arrowsize=0.1
  2.0,arrowlength=2,arrowinset=0.4]{->}(1.5884376,-0.3)(3.1884375,0.1)
\psline[linewidth=0.04cm](3.1884375,0.1)(1.5884376,-1.1)
\psline[linewidth=0.04cm](6.3884373,-1.1)(4.7884374,-1.1)
\psline[linewidth=0.04cm](6.3884373,0.1)(4.7884374,0.1)
\psline[linewidth=0.04cm](6.3884373,1.3)(4.7884374,1.3)
\psline[linewidth=0.04cm](4.7884374,1.3)(3.1884375,0.1)
\psline[linewidth=0.04cm](3.1884375,0.1)(4.7884374,0.1)
\psbezier[linewidth=0.04,arrowsize=0.1
  2.0,arrowlength=2,arrowinset=0.4]{->}(6.3884373,1.3)(5.3884373,2.1)(3.7884376,1.9)(3.1884375,0.1)
\psline[linewidth=0.04cm](3.1884375,0.1)(4.7884374,-1.1)
\psbezier[linewidth=0.04,arrowsize=0.1
  2.0,arrowlength=2,arrowinset=0.4]{->}(6.3884373,0.1)(4.7884374,0.5)(4.5884376,0.5)(3.1884375,0.1)
\psbezier[linewidth=0.04,arrowsize=0.1
  2.0,arrowlength=2,arrowinset=0.4]{->}(6.3884373,-1.1)(4.7884374,-2.1)(3.5884376,-1.1)(3.1884375,0.1)
\rput(2.4628124,1.0){\small 2}
\rput(2.0628126,0.6){\small 2}
\rput(2.0628126,0.0){\small 2}
\rput(4.2628126,1.8){\small 2}
\rput(5.2628126,0.6){\small 2}
\rput(3.8628125,-1.4){\small 2}
\rput(0.89234376,-1.0){\small $c$}
\rput(0.89234376,-1.4){\small $-(n-k)$}
\rput(0.64234376,-0.2){\small $c_{i,1}$ $(-1)$}
\rput(0.64234376,0.6){\small $c_{i,2}$ $(-1)$}
\rput(0.64234376,1.4){\small $c_{i,3}$ $(-1)$}
\rput(5.38594,1.0){\small $d_{i,1}$ $(-1)$}
\rput(4.8485937,-0.2){\small $d_{i,2}$ $(-1)$}
\rput(5.248594,-0.8){\small $d_{i,3}$ $(-1)$}
\rput(7.2785935,1.2){\small $d'_{i,1}$ $(-1)$}
\rput(7.2785935,0.0){\small $d'_{i,2}$ $(-1)$}
\rput(7.2785935,-1.2){\small $d'_{i,3}$ $(-1)$}
\rput(3.010625,-0.4){\small $S_{i}$}
\rput(2.9810936,-0.8){\small (3)}
\end{pspicture} 
}
\caption{Gadget used in the Copeland$^1$ manipulation NP-hardness proof~\cite{fal-hem-sch:c:manip-copeland}: The gadget is constructed for each $S_{i}$ with numbers for each candidate showing the lead in Copeland score the preferred candidate has over them.}
\end{figure}
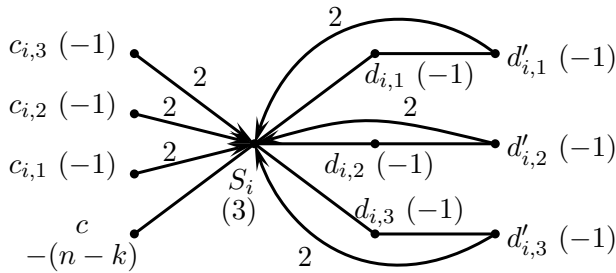

Given any instance of an X3C problem, the proof defines the set of
candidates as the elements of sets $D$ and $\cal S$, one preferred
candidate $p$, and one main adversary candidate $c$ (the candidate
holding the highest Copeland score prior to manipulation) as well as
auxiliary candidates used for padding and constructing the
mathematical gadgets.  Such gadgets (also called
widgets~\cite{cor-lei-riv-ste:b:algorithms-second-edition}) are a
common feature in reduction constructions: They are typically
intricate subgraphs that are constructed for each element of the
mapped-from problem.  They align calculations and enforce constraints
so as to ensure a tight mapping between instances of the two problems.

The numbers of the election are constructed in such a way that prior
to manipulation, the elements of $D$ each beat $p$ by 1 unit of Copeland
score and $c$ beats $p$ by $n-k$ units of Copeland score. However, the
elements of $\cal S$ each lose to $p$ by 3 units. Hence any successful set
of manipulating votes must cause the elements of $D$ and $c$ to lose
just enough pairwise contests against the elements of $\cal S$ to
erode their lead over $p$ without making any $S_{i}$ a possible
winner. This involves selecting a subset of $\cal S$ that beats every
candidate from $D$ but still leaves enough elements of $\cal S$ to lower $c$'s
score by $n-k$ points. The construction of such a vote corresponds to
selecting a $k$-sized subset of $\cal S$ that covers exactly the
elements of $D$, that is, a solution to the X3C problem. Hence, a
solution to UCM$_{Cope^{\alpha}}$ gives us a solution to X3C.
\end{proof}



\section{UCM in Second-Order Copeland}
Solutions for breaking ties in Copeland voting attempt to choose the
more ``powerful'' candidate as the winner, which can be defined in a
number of ways. One such method is second-order Copeland. It selects the candidate whose set
of defeated opponents (hereafter referred to as $DO_{c}$ for a
candidate $c$) has
the higher sum of Copeland scores. We will let $sum\_score(S)$ for a
set of candidates $S$ be the sum of the Copeland scores for candidates in
$S$, and so $sum\_score(DO_c)$ gives the second-order Copeland score
for a candidate $c$.

 Second-order Copeland has been used
by the National Football League and the United States Chess Federation
to break ties, and has a special place in voting theory: it was the
first tractable voting system for which the manipulation problem  was
shown to be NP-complete, even for just one manipulator~\cite{bar-tov-tri:j:manipulating}.

\subsection{Second-Order Copeland is Resistant to Manipulation}
The problem of unweighted coalitional manipulation in second-order
Copeland, hereafter referred to as UCM$_{2Cope}$, is NP-complete even
for one manipulator~\cite{bar-tov-tri:j:manipulating}. Verifying a
given solution is clearly in P as we simply have to calculate
Copeland scores and second-order Copeland scores for each
candidate. To prove UCM$_{2Cope}$ is NP-hard we show a polynomial-time
reduction from 3,4-SAT, a known NP-hard
problem~\cite{tov:j:sat-repeats}.

\noindent
\textbf{3,4-SAT}
\begin{description}
\item[Given] A set $U$ of Boolean variables, a collection of clauses
  $Cl$, each clause composed of disjunctions of exactly three
  \emph{literals}, which may be a variable or its complement, and each
  variable occurs in exactly four clauses.
  
\item[Question] Does there exist a Boolean assignment over $U$ such
  that each clause in $Cl$ contains at least one literal set to true?
\end{description}

To facilitate the reduction we construct a graph representation of a
second-order Copeland election that encodes the given 3,4-SAT
instance.  We will use a election graph representation with vertices
representing candidates and directed edges representing the result of
pairwise contests.

\begin{proof}

Given a 3,4-SAT instance, we create a
second-order Copeland election graph as follows: Every clause ($C_{1}$
to $C_{||Cl||}$) and every literal is a
candidate, represented as a vertex in our graph.  
The manipulating coalition's
chosen candidate is a separate candidate $C_{0}$. All pairs
of vertices have directed edges between them, representing decided
pairwise contests, except for any variable and its complement. The
decided contests cannot be overturned by our manipulators, while undecided contests can
be shifted in either direction according to the manipulating vote.  Clauses
beat (that is, have a directed edge to) literals they contain,
and lose to all other literals. 


In addition to these candidates derived from the 3,4-SAT instance, we pad
the election with a number of auxiliary candidates in such a way to
achieve the desired Copeland scores and second-order Copeland scores
for each of the candidates.  We will have that all the clause
candidates and $C_{0}$ are tied with the highest Copeland score.  We
will also have that each clause candidate $C_{i}$ has
$sum\_score(DO_{C_{i}}) = sum\_score(DO_{C_{0}}) -3$.  


The second-order Copeland score for
$C_{0}$ will be  independent of the variable-complement contests.  For all
possible outcomes of the variable-complement contests, $C_{0}$ still
beats every candidate except for the clause candidates.  Recall that the elements
for every clause candidate's defeated-opponent set are their component
literals.  Their second-order Copeland scores and the final result of
the election will then depend on how each of the variable-complement
contests are decided.



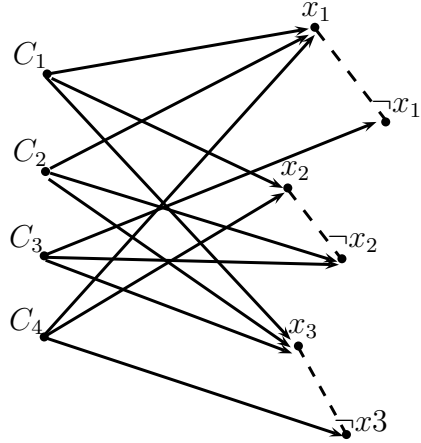
\begin{figure}
\centering
\scalebox{1} 
{}
\begin{pspicture}(0,-2.94125)(5.295625,2.96125)
\psdots[dotsize=0.12](0.2796875,1.93875)
\psdots[dotsize=0.12](0.2596875,0.63875)
\psdots[dotsize=0.12](0.2396875,-0.48125)
\psdots[dotsize=0.12](0.2396875,-1.56125)
\psdots[dotsize=0.12](3.8396876,2.55875)
\psdots[dotsize=0.12](4.7796874,1.29875)
\psdots[dotsize=0.12](3.4796875,0.41875)
\psdots[dotsize=0.12](4.1996875,-0.52125)
\psdots[dotsize=0.12](3.6196876,-1.68125)
\psdots[dotsize=0.12](4.2596874,-2.86125)
\usefont{T1}{ppl}{m}{n}
\rput(0.0575,2.18875){$C_{1}$}
\usefont{T1}{ppl}{m}{n}
\rput(0.0575,0.88875){$C_{2}$}
\usefont{T1}{ppl}{m}{n}
\rput(0.0075,-0.25125){$C_{3}$}
\usefont{T1}{ppl}{m}{n}
\rput(0.0075,-1.35125){$C_{4}$}
\usefont{T1}{ppl}{m}{n}
\rput(3.8625,2.76875){$x_{1}$}
\usefont{T1}{ppl}{m}{n}
\rput(4.9325,1.48875){$\lnot x_{1}$}
\usefont{T1}{ppl}{m}{n}
\rput(3.5810938,0.64875){$x_{2}$}
\usefont{T1}{ppl}{m}{n}
\rput(4.351094,-0.31125){$\lnot x_{2}$}
\usefont{T1}{ppl}{m}{n}
\rput(3.6846876,-1.45125){$x_{3}$}
\usefont{T1}{ppl}{m}{n}
\rput(4.4546876,-2.67125){$\lnot x3$}
\psline[linewidth=0.04cm,arrowsize=0.05291667cm 2.0,arrowlength=1.4,arrowinset=0.4]{->}(0.2996875,1.93875)(3.7596874,2.53875)
\psline[linewidth=0.04cm,arrowsize=0.05291667cm 2.0,arrowlength=1.4,arrowinset=0.4]{->}(0.3396875,1.87875)(3.4196875,0.41875)
\psline[linewidth=0.04cm,arrowsize=0.05291667cm 2.0,arrowlength=1.4,arrowinset=0.4]{->}(0.2596875,1.89875)(3.5196874,-1.60125)
\psline[linewidth=0.04cm,arrowsize=0.05291667cm 2.0,arrowlength=1.4,arrowinset=0.4]{->}(0.3196875,0.67875)(3.7596874,2.45875)
\psline[linewidth=0.04cm,arrowsize=0.05291667cm 2.0,arrowlength=1.4,arrowinset=0.4]{->}(0.3196875,0.61875)(4.1396875,-0.56125)
\psline[linewidth=0.04cm,arrowsize=0.05291667cm 2.0,arrowlength=1.4,arrowinset=0.4]{->}(0.2996875,0.53875)(3.5196874,-1.70125)
\psline[linewidth=0.04cm,arrowsize=0.05291667cm 2.0,arrowlength=1.4,arrowinset=0.4]{->}(0.2796875,-0.46125)(4.6596875,1.29875)
\psline[linewidth=0.04cm,arrowsize=0.05291667cm 2.0,arrowlength=1.4,arrowinset=0.4]{->}(0.2196875,-0.50125)(4.1596875,-0.60125)
\psline[linewidth=0.04cm,arrowsize=0.05291667cm 2.0,arrowlength=1.4,arrowinset=0.4]{->}(0.2596875,-0.54125)(3.5396874,-1.78125)
\psline[linewidth=0.04cm,arrowsize=0.05291667cm 2.0,arrowlength=1.4,arrowinset=0.4]{->}(0.2796875,-1.58125)(4.2196875,-2.88125)
\psline[linewidth=0.04cm,arrowsize=0.05291667cm 2.0,arrowlength=1.4,arrowinset=0.4]{->}(0.2396875,-1.56125)(3.4396875,0.35875)
\psline[linewidth=0.04cm,arrowsize=0.05291667cm 2.0,arrowlength=1.4,arrowinset=0.4]{->}(0.2396875,-1.54125)(3.8396876,2.45875)
\psline[linewidth=0.04cm,linestyle=dashed,dash=0.16cm 0.16cm](3.8596876,2.51875)(4.7996874,1.29875)
\psline[linewidth=0.04cm,linestyle=dashed,dash=0.16cm 0.16cm](3.4796875,0.39875)(4.2196875,-0.56125)
\psline[linewidth=0.04cm,linestyle=dashed,dash=0.16cm 0.16cm](3.6196876,-1.70125)(4.2596874,-2.86125)
\end{pspicture} 
\caption{A partial representation of the resultant
  election graph: Clauses represented are $C_{1} (x_{1} \vee x_{2} \vee x_{3})$, $C_{2} (x_{1} \vee \lnot x_{2} \vee x_{3})$, $C_{3} (\lnot x_{1} \vee \lnot x_{2} \vee x_{3})$ and $C_{4} (x_{1} \vee x_{2} \vee \lnot x_{3})$. Each clause vertex beats the variables (or their
  complements) that are its component literals. The dotted lines indicate
  undecided (second-order Copeland) contests between variables and
  their complements.}
\end{figure}

Consider if any clause candidate
$C_{i}$'s literals win all their contests. The $sum\_score(DO_{C_{i}})$
increases by 3 points and $C_{i}$ is tied with
$C_{0}$ for first place. Therefore, in order for $C_{0}$ to be
the unique winner, at least one element of each clause candidate's defeated-opponent set must lose
one of their contests.  
 For any variable $x$, we can interpret a directed edge from $x$ to
 $\lnot x$ as setting $x$ to $true$.  A vote that allows $C_{0}$ to
 win, then, would correspond exactly to each clause having at least
 one literal evaluating to true. In other words, we have a solution
 for UCM$_{2Cope}$ if and only if we have solution for
 3,4-SAT\@. Therefore UCM$_{2Cope}$ is NP-hard with just a single
 manipulator.
\end{proof}


\section{UCM in Maximin}
Maximin voting, also known as the Simpson-Kramer method, is
a typical Condorcet voting system. As such it deals with contests
between pairs of candidates, specifically their $netadv$
scores. To find the winner under maximin, given a $netadv$
function over the set of candidates, we first select the
lowest $netadv$ score for each candidate $k$ in $C$, i.e., we select
the minimum score for $netadv(k, k')$ for all $k'$ in $C$ such that
$k' \neq k$. The winner is the candidate with the highest such
score. We can trivially see that a candidate with a
minimum $netadv$ score greater than $0$ will be the Condorcet winner
and there can only be one such candidate, thus maximin is a Condorcet
voting system.

Calculating the maximin winner is easily seen to be polynomial in the
size of the election, but Xia et
al.~\cite{con-pro-ros-xia:t:unweighted-manipulation} prove that for
two or more manipulators the problem of UCM$_{maximin}$ is
NP-complete.


\subsection{Maximin is Resistant to Manipulation}
UCM$_{maximin}$ is NP-complete for two or more
manipulators~\cite{con-pro-ros-xia:t:unweighted-manipulation}: Verifying
an instance is easily seen to be polynomial in the size of the election as we can calculate the winner in polynomial time.  To
prove UCM$_{maximin}$ is NP-hard we construct a polynomial-time
reduction from the \emph{vertex-disjoint-two-path} problem, known to
be NP-complete~\cite{LaPaugh:1978:SHP:800133.804330}.

\noindent
\textbf{\boldmath Vertex-Disjoint-Two-Path problem (VDP$_{2}$)}
\begin{description}
\item[Given] A directed graph $G$ and two sets of vertices $u,u'$ and
  $v,v'$ such that all four vertices are unique.
\item[Question] Do there exist two paths $u \rightarrow u_{1}
  \rightarrow \ldots \rightarrow u_{j} \rightarrow u'$ and $v
  \rightarrow v_{1} \rightarrow \ldots \rightarrow v_{k} \rightarrow
  v'$ in $G$ such that each path is a set of vertices disjoint from
  the other?
\end{description}

\begin{proof}
To facilitate our reduction construction from a graph problem such as
the vertex-disjoint-two-path problem to maximin, we first construct a
graph representation of maximin elections.  A complete set of votes is
not required to represent an election\footnote{Refer Appendix A.}. We
can construct the same given just a $netadv$ (or $adv$)
function. Also, there exists a bijection between the $netadv$ function
and directed edges of a complete antisymmetric graph such that given one,
we can represent it in terms of the other\footnote{Refer Appendix B.}.



$P_{coal}$ is the set of votes of the manipulating
coalition and $P_{noncoal}$ is the set of all other
votes. The term $netadv_{noncoal}$ indicates the $netadv$ score obtained by
considering only the noncoalitional votes with $netadv_{coal}$
similarly defined. $M$ is the size of the coalition and $c$ is the
candidate supported by the manipulating coalition.

Given a VDP$_2$ instance, that is a graph $G(V_G, E)$ and vertices
$u,u',v,v' \in V$, we obtain a
graph $G'$ using the following constructions and assumptions:

\begin{itemize}
\item  Every vertex in our graph is reachable from $u$ or $v$.
\item  There are no directed edges $u \rightarrow v' $ or $v
  \rightarrow u' $.
\item  We add special edges $u' \rightarrow v$ and $v' \rightarrow
  u $ such that $E_{G'} = E \cup \{(u',v), (v', u)\}$.
\end{itemize}

Our UCM$_{maximin}$ instance then is as follows:
\begin{itemize}
\item Set of candidates $C = V_G$.
\item 
  The set of $P_{noncoal}$ preferences.
  \begin{itemize}
  \item $\forall c' \in C: c' \neq c, netadv(c,c') = -4M $;
  \item $netadv(u,v') = netadv(v,u') = -4M$;
  \item For all other edges $(x,y)$ in $E$, $netadv(y,x) = -2M - 2$;
  \item For all other vertex pairs $a,b$, $netadv(a,b) = 0$.
  \end{itemize}
\end{itemize}

Regarding $P_{coal}$ votes, we can freely assume that every coalition
vote will rank $c$ first, thus giving $netadv_{coal}(c,c') = M$ and
$netadv_{noncoal \cup coal}(c,c') = -3M$. Thus, in our construction,
scores for $netadv(c,c')$ are fixed for all candidates $c' \in C$.  In
order for $c$ to be the winner, at least one $netadv$ score must be
less than $-3M$ for every other candidate. We will see that in our
construction this can occur if and only if there are two
vertex-disjoint paths in $G'$.

\noindent
\textbf{Direction 1} The existence of vertex-disjoint paths $u \rightarrow u_{1}
\rightarrow \ldots \rightarrow u_{j} \rightarrow u'$ and $v
\rightarrow v_{1} \rightarrow \ldots \rightarrow v_{k} \rightarrow
v'$  yields a $P_{coal}$
that makes $c$ the winner. In order to construct the manipulative
preferences, we will make use of a connected subgraph over $G'$
containing all the vertices, but with $u \rightarrow
\ldots \rightarrow u' \rightarrow v \rightarrow \ldots \rightarrow v'
\rightarrow u $ as the only cycle in the graph.  


We can construct $P_{coal}$ votes in 3 parts as follows:
Each manipulator vote will rank $c$ the highest, followed by the
vertex-disjoint-path vertices, followed by the other vertices.

\textit{Other-vertex ordering:} 
These vertices will be ordered in the votes based on a linear order
extracted from the single-cycle subgraph.

\textit{Vertex-disjoint-path orderings:} We have two
vertex-disjoint-path orderings: $u \rightarrow \ldots \rightarrow u'
\rightarrow v \rightarrow \ldots \rightarrow v'$ and $v \rightarrow
\ldots \rightarrow v' \rightarrow u \rightarrow \ldots \rightarrow u'$
and thus two possible vote constructions for $P_{coal}$. We construct
$M-1$ votes as per the first ordering and 1 vote as per the
second\footnote{Switching these orderings results in the same
  outcome.}. Thus $netadv(c,c')$ increases by $M$ points but every
other $netadv$ score increases by less than $M$ points, making $c$ the
winner. The calculations are as follows for the complete (coalitional
and noncoalitional) set of votes:

\begin{itemize}
\item $netadv(u,v') = -4M + (M -1) - 1 = -3M -2$ 
\item $netadv(v,u') = -4M + 1 -(M -1) = -5M +2$ 
\end{itemize}

For any other candidate $c' \not\in \{c,u,v\}$, we can see that there
exists some candidate $d$ in every vote of $P_{coal}$ that beats
$c'$, i.e., the lowest $netadv$ score for $c'$ is \mbox{$netadv_{coal}(c',d)
= -M$}, and thus for the complete set of votes the (lowest) $netadv$
for any such candidate $c'$ is no more than $-2M -2 -M = -3M -2$. All the above
$netadv$ scores are less than $-3M$ for all values of $M \geq 2$, thus
$c$ is the winner.

\noindent
\textbf{Direction 2} The existence of a $P_{coal}$ that makes $c$ a
winner yields a positive  VDP$_{2}$ instance in the graph $G'$:

Since $c$ is the winner, we know that for any other candidate $c'$ in
$C$, there exists a candidate $d$ that beats $c'$ such that:

\begin{itemize}
\item $netadv(c',d) < -3M$;
\item There exists an edge $(d, c')$ in $G'$;
\item $d$ is ranked higher than $c'$ in a majority of the total votes
  and in at least one vote in $P_{coal}$---the proof of this
  is as follows.
\end{itemize}

Consider such an edge $(d, c')$\footnote{If there exists more than one
  such $d$ we choose one arbitrarily.}: either $(d, c')$ is one of the
special edges $(v', u)$, $(u', v')$ or $(d, c') \in E$. If $(d, c')$
is a special edge, then at least one vote in $P_{coal}$ must prefer
$d$ to $c'$ (since $netadv(c',d) < -3M$). If $(d, c') \in E$, then all
$M$ votes in $P_{coal}$ must prefer $d$ to $c'$.

For this $d$, we can choose a candidate that beats it with sufficient
margin, and continue to find such a candidate for the previous choice
of $d$\footnote{Choosing such a $d$ is formalized in the proof of Xia
  et al.~\cite{con-pro-ros-xia:t:unweighted-manipulation} as a
  composite function $f$.}. That is, we find a $d$ for $c'$ starting with
$u$ and then continue to find such a candidate after setting $d$ to
$c'$. There is only one possible $d$ for $c'$ as either $u$ or
$v$. Thus, we obtain a set of chained pairs recursively. This, coupled
with the facts that any vertex is reachable from $u,v$ and the
existence of special edges (both by construction), we obtain a cycle
of vertices which breaks into disjoint sets along the special
edges. Obtaining such a set of pairs with $netadv$ less than $-3M$ is
not possible without the existence of a ($c$-winner-making) $P_{coal}$
(as per the third item above).
\end{proof}

\section{UCM in Tideman Ranked Pairs}
Tideman ranked pairs (TRP) was conceived by Nicolaus Tideman in 1987
when attempting to define a voting system that ``almost always'' has
the property of independence of
clones~\cite{springerlink:10.1007/BF00433944}.  It is defined as
follows: given a $netadv$ function over the set of candidates, create
a list by ranking the pairs in descending order of their scores.  In
the case of a tie between two $netadv$ pairs, e.g., $netadv(a,b) =
netadv(x,y)$, we break ties by ordering the pairs lexicographically
according to an arbitrary ordering of the candidates.  We add the
top-scoring pair to an election graph $G$\footnote{As usual, $V_G = C$
  and directed edges represent $netadv$ scores.} if the resultant graph does not contain a cycle. Otherwise we
skip this pair and move on to the next one in the list.  We continue
until we have considered all pairs.  Since we now have a directed acyclic
graph, there must exist a source vertex, which we state to be the TRP
winner.

\subsection{TRP is Resistant to Manipulation}
Xia et al.~\cite{con-pro-ros-xia:t:unweighted-manipulation} found that
UCM$_{TRP}$ is NP-complete even for one manipulator. We can easily see
that verifying a given solution to UCM$_{TRP}$ is in P\@. UCM$_{TRP}$
was proven to be NP-hard by a polynomial-time reduction from 3SAT.

\noindent
\textbf{Three-Conjunctive-Normal-Form Satisfiability (3SAT)}

\begin{description}
\item[Given] A set $U$ of Boolean variables,
  a collection of clauses $Cl$, each clause composed of disjunctions
  of exactly three \emph{literals}, which may be a variable or its
  complement.  
\item[Question] Does there exist a Boolean assignment over $U$ such
  that each clause in $Cl$ contains at least one literal set to true?
\end{description}

\begin{proof}
Given a 3SAT instance, we construct a UCM$_{TRP}$ election graph as
follows. Clauses $C_{1} \ldots C_{||Cl||}$ are vertices as is the
coalition's preferred candidate $c$. For each clause $C_{i}$, we
construct six other special clause candidates---three for the literals
it contains and three for their complements. $C_{i}$
beats (has a directed edge to) the
special clause candidates corresponding to the literals it contains,
which in turn beat the literals they correspond to.  We also have a
$C_{i}'$ such that it is beaten by the complement of the literals in
$C_{i}$. Candidate $c$ starts out beating the $C_{i}$ candidates but
gets defeated (though by a smaller margin) by the $C_{i}'$
candidates. The intuition of this correspondence is that the edges beating
$c$ are so weak and so far down the ordering that they are not added,
leaving $c$ to be the source vertex (and TRP winner) if and only if there exists
a solution to 3SAT\@. Thus the UCM problem is NP-complete even in the
case of a single manipulator. The proof of Xia et al.~\cite{con-pro-ros-xia:t:unweighted-manipulation} relies
on mathematical gadgets to achieve this correspondence.
\end{proof}

\section{Conclusion}

Our survey of UCM results can be seen as qualifying election systems
by a single metric. Determining which election system is superior is
an ongoing debate often reflecting differing
philosophies. Pierre-Simon Laplace in his lectures at the Ecole
Normale Superieure in 1795 attacked the \emph{\'{e}l\'{e}ction par ordre de m\'{e}rite}
(election by ranking of merit) system of his contemporary Jean-Charles
de Borda~\cite{szp:b:numbers}, later proposing a variation of the
majority rule in its place. Much of the modern literature on voting
theory is still devoted to advocacy for particular voting systems,
arguing their superiority by  one metric or another~\cite{newenhizen:bordarespectscondorcet,saari:Bordabetter,risse:cordorcetborda}.


In our survey we showcase manipulation results for a
particular class of voting systems, namely those with a tractable
winner problem but where unweighted coalitional manipulation is hard
for a constant coalition size. The complexity of this case of  UCM  has been determined for
most common voting rules, though a few remain: Copeland$^{0.5}$
remains unsolved even as results for all other parameter values have
been found~\cite{fal-hem-sch:c:manip-copeland}.

\begin{table}[t]\footnotesize
\begin{tabular}{|c|c|c|}
\hline
Voting rule & Coalition size = 1 & Coalition size $\geq$ 2\\
\hline
Copeland$^{\alpha} (0 < \alpha < 0.5)(0.5 < \alpha < 1)$ & P~\cite{bar-tov-tri:j:manipulating} & NP-complete~\cite{fal-hem-sch:c:copeland-ties-matter}\\
Copeland$^{\alpha}(\alpha = \{0,1\})$  & P~\cite{bar-tov-tri:j:manipulating} & NP-complete~\cite{fal-hem-sch:c:manip-copeland}\\
Copeland$^{\alpha}(\alpha = \{0.5\})$  & P~\cite{bar-tov-tri:j:manipulating} & ? \\
Second-order Copeland & NP-complete~\cite{bar-tov-tri:j:manipulating} & NP-complete~\cite{bar-tov-tri:j:manipulating}\\
Single Transferable Vote & NP-complete~\cite{bar-oli:j:polsci:strategic-voting} & NP-complete~\cite{bar-oli:j:polsci:strategic-voting}\\
Maximin & P~\cite{bar-tov-tri:j:manipulating} & NP-complete~\cite{con-pro-ros-xia:t:unweighted-manipulation}\\
Tideman Ranked Pairs & NP-complete~\cite{con-pro-ros-xia:t:unweighted-manipulation} & NP-complete~\cite{con-pro-ros-xia:t:unweighted-manipulation}\\
Borda & P~\cite{bar-tov-tri:j:manipulating} & NP-complete~\cite{bet-nie-woe-c-borda, dav-kat-nar-wal:j:borda}\\
Bucklin & P~\cite{con-pro-ros-xia:t:unweighted-manipulation} & P~\cite{con-pro-ros-xia:t:unweighted-manipulation}\\
Plurality with Runoff & P~\cite{pro-ros-zuc:j:borda} & P~\cite{pro-ros-zuc:j:borda}\\
Veto & P~\cite{bar-tov-tri:j:manipulating} & P~\cite{pro-ros-zuc:j:borda}\\
Cup & P~\cite{con-lan-san:j:few-candidates} & P~\cite{con-lan-san:j:few-candidates}\\

\hline
\hline
\hline
\end{tabular}
\caption{Table of UCM results for common tractable voting systems.}
\label{tab:template}
\end{table}


Several related areas of research, however, remain more or less
uncharted. The most significant simplification in the literature is
that most of the hardness results achieved are just worst-case.
Several papers have studied whether voting systems are difficult to
manipulate in a large fraction of instances, finding that manipulation
can be easy in the average case while being hard in the worst
case~\cite{con-san:c:nonexistence,pro-ros:j:juntas,pro-ros:c:average}.
Additionally, approximation algorithms exist for several worst-case
hardness results.  Brelsford et
al.~\cite{bre-fal-hem-sch-sch:c:approximating-elections} formalized
manipulation as an optimization problem and then studied whether this
version of the problem is approximable.  Zuckerman et
al.~\cite{pro-ros-zuc:j:borda} discovered an approximation algorithm for
Borda manipulation before it was known to be NP-hard, and Davies et
al.~\cite{dav-kat-nar-wal:j:borda} gave several other approximation algorithms for this
problem.
In other results, approximation algorithms for manipulation of 
maximin as well as families of scoring protocols exist~\cite{Zuckerman:2010:COMSOC,Xia_ascheduling}. 
Other techniques include the use of relatively efficient algorithms for the
NP-complete integer partitioning problem to solve manipulation
instances~\cite{lin:c:partition}.

Conitzer et al.~\cite{con-lan-san:j:few-candidates} qualified the manipulation problem with an additional
metric: the minimum number of candidates that must be present for
manipulation to be
NP-hard.  Additionally their work
breaks from the standard model and studies whether manipulation is
hard for cases where manipulators do not have complete information of
all of the votes.  Slinko explored how often elections will be
manipulable based on the size of the manipulative
coalition~\cite{sli:j:howlarge}.

Research into how often elections can be
manipulated~\cite{fri-kal-nis:c:quantiative-gib-sat}, and more general
areas such as parametrization of NP-hard
problems~\cite{niedermeier:LIPIcs:2010:2495} and phase
transitions~\cite{Cheeseman91wherethe, Kirkpatrick27051994,
  phasetranzhang} lead to a more nuanced approach to problem
classification.  Phase transitions have been examined in the
manipulation problem
for the veto rule~\cite{wal:c:phase}.


Another issue is that votes are most commonly represented in the
literature as transitive linear preference orderings over the
set of candidates and the concept of irrational votes has only been
sparsely dealt
with. Irrational (by which
we mean intransitive) votes, may be more apt for any number of
real-world scenarios where voters tend to rank candidates according to
multiple criteria. Irrational votes are not represented as a linear
ordering but as a preference table which holds the voter's choice
for any pair of candidates.   
For Copeland$^\alpha$ for $\alpha \in
 \{0,0.5,1\}$, manipulation is in P in the irrational voter model,
while it is known to be NP-hard for $\alpha \in \{0,1\}$ in the
standard voter model~\cite{fal-hem-sch:c:manip-copeland}.  Thus voting
systems may have different behavior with regard to manipulation in the
irrational voter model and it deserves more study.  Another convention
is that the default definition of UCM is constructive---i.e., efforts
are directed to making a preferred candidate a winner, rather than
preventing a certain candidate from winning. Variations of UCM with a
destructive approach is another area rich with possibilities.

UCM instances presented in this paper typically have a large number of
candidates and a smaller constant-sized coalition of manipulators. In
contrast, cases with a small number of candidates and a relatively
large manipulating coalition might be considered more natural. Betzler
et al.~\cite{bet-nie-woe-c-borda} mention a specific open problem in this area: whether
there exists a combinatorial algorithm to solve Borda efficiently with
few candidates and an unbounded coalition
size. Another open problem is solving a
UCM$_{Borda}$ instance having a coalition of size 2 in less than
$O(||C||!)$. 

Another approach to the manipulation problem taken by Conitzer and
Sandholm~\cite{con-san:c:voting-tweaks} and Elkind and
Lipmaa~\cite{elk-lip:c:small} is modifying voting systems to give them
greater resistance to manipulation.  Both add an extra initial round
of subelections between subsets of the candidates.  Conitzer and
Sandholm~\cite{con-san:c:voting-tweaks} describe techniques that can make manipulation NP-hard or
even PSPACE-hard for these modified voting systems.  Elkind and Lipmaa~\cite{elk-lip:c:small}
present a version of this technique that uses one-way functions to
construct the initial-round schedule from the set of votes.  Reversing
the one-way function is computationally hard, preventing election
organizers from gaming the initial round and forcing their desired result
in polynomial time.  These techniques essentially construct new voting
systems by structurally augmenting standard systems to imbue them
with complexity.

Other related work includes the study of electoral control, which
encompasses \mbox{attempts} by an election organizer to change the
result by modifying the election structure in various
ways~\cite{bar-tov-tri:j:control,fal-hem-hem-rot:j:llull,erd-fel-pir-rot:c:param,hem-hem-rot:j:hybrid},
which also encompasses cloning, or adding candidates very similar to
existing candidates in an attempt to split their
support~\cite{springerlink:10.1007/BF00433944,elk-fal-sli:c:clone}.
Other ways to influence elections include bribery and campaign
management, where in both cases a briber attempts to sway the result
of an election by paying off a set of voters to change their
votes~\cite{fal-hem-hem:j:bribery,fal-hem-hem-rot:j:llull,fal:j:nonuniform,elk-fal-sli:c:swap,elk-fal:j:approx-camp,sch-fal-elk:c:camp}.
These, too, are problems endemic to many voting systems to which
complexity can serve as a defense.


Another possible response to the problems presented by Arrow's theorem
and the Gibbard-Satterthwaite theorem is to reconsider the standard
model of the aggregate function.  Balinski and Laraki~\cite{bal-lar:c:theory} introduce a
model where voters give candidates independent grades, such as the
letter grades F to A or $\{$good, average, bad$\}$, similar to
approval voting or range voting, rather than ranking them in a linear
order.  In a sense this represents a reversion
to the pre-Bergson-Samuelson model of welfare functions.
Balinski and Laraki's method defines the aggregate grade of each
candidate to be the median grade over all votes, unlike range voting
where the aggregate grade is the average.  We can obtain a complete
aggregate \mbox{preference} ordering of the candidates with this
method provided that ties can be broken. Balinski and Laraki give a
tie-breaking mechanism that successively removes one of the
median-score-awarding voters from the votes for each tied candidate
and recomputes the median grades until they are no longer
tied~\cite{bal-lar:c:theory}.  Their approach does not
rely on complexity but instead redesigns the election model to become
strategy-proof in a limited case defined by the authors.







Faliszewski et al.~\cite{fal-hem-hem-rot:j:shield} showed
that with a restriction to single-peaked preferences, a wide range of manipulation and
control instances that are NP-hard in the general case turn out to be
easy (though not any of the results we describe here).  For some
voting systems these problems remain easy even with a partial relaxation of
the single-peaked model that allows for a small number of ``mavericks'',
whose votes are not aligned with  the single-peaked ordering~\cite{fal-hem-hem:j:nearly}.  The
single-peaked model is considered ``\emph{the} canonical setting for
models of political institutions''~\cite{gal-pat-pen:j:arrow}, so
this work calls the significance of a number of hardness results into
question.  


After the birth of research in the manipulation problem with the work
of Bartholdi et al.~\cite{bar-tov-tri:j:manipulating}, most research
moved towards the weighted voter model and many results for the
weighted coalitional manipulation problem (WCM) were
achieved~\cite{con-lan-san:j:few-candidates, hem-hem:j:dichotomy},
until the resurgence of interest in the UCM
problem~\cite{bet-nie-woe-c-borda,
  dav-kat-nar-wal:j:borda,fal-hem-sch:c:manip-copeland,
  fal-hem-sch:c:copeland-ties-matter,con-pro-ros-xia:t:unweighted-manipulation}.
It can be argued that as compared to WCM, UCM is a better test of a
voting system's vulnerability to manipulation. UCM serves as a special
case of WCM and hence subsumes its hardness results. In other words,
if an election system is resistant to manipulation in the UCM case, it
will resist manipulation in the WCM case, but the other direction does
not necessarily follow.  With this problem solved for most common
voting systems, we look forward to the resolution of the remaining
open problems as well as new avenues of research into the
manipulability of voting systems.

\appendix

\section{Constructing an Election Given a  {\boldmath $netadv$} Function: the McGarvey Method}

While the traditional representation of an election requires a set of
votes, we can construct these votes given a pairwise relation denoting
preference over the
set of candidates.  This method was given by
McGarvey~\cite{mcg:j:election-graph}, and can be applied with very
little modification to a $netadv$ function.

\subsection{From a Preference Pattern to a Set of Votes} 
\begin{theorem}

\emph{(McGarvey's Theorem)}
\label{McGarvey}
Given a \emph{preference pattern} we can elicit a set of votes
(defined to be strict and complete preference orderings over the set
of candidates) such that (the ordering derived from) the preference
pattern is the result of the election.
\end{theorem}
\begin{definition}
A \emph{preference pattern} is a set of relations over the
set of candidates. The relations are a preference relation
(expressed as $aPb$ viz.~$a$ is preferred to $b$) and an indifference
relation ($aIb$ viz.~$a$ is neither preferred to $b$ nor is $b$
preferred to $a$).
\end{definition}

Both relations are \emph{distinct} for any pair of candidates -
i.e., $aPb$ implies $\lnot bPa$, and $aIb$ implies $bIa$.  Thus we will
have $m(m-1)/2$ pairs over both relations where $m$
is the number of candidates. McGarvey's method
constructs a set of votes as follows:

For each pair $aPb$ with remaining candidates $c_1,\dots,c_{m-2}$, we
construct two preference orderings $abc_{1} \ldots c_{m-2}$ and
$c_{m-2} \ldots c_{1} a b$\footnote{Where candidates are listed in
  order of decreasing preference.}. For each pair $aIb$, we construct
$a b c_{1} \ldots c_{m-2}$ and $c_{m-2} \ldots c_1 b a$ The idea is
that on evaluation for these preferences, rankings of all candidates
besides $a,b$ from these orderings will be equal, and the rankings of
$a,b$ reflect the preference relation under consideration.  Consider
the example $C = {a,b,c,d}$. The six pairs we consider are ${aPb, aPc,
  aPd, bPc, bPd, cPd}$.

To represent $aPb$ we construct two votes $abcd$ and
$dcab$. Evaluating these two votes in the context of pairwise rankings
leads to two votes for $aPb$ and no votes for any other pair over the
set of candidates. Similarly, for $aPc$, we construct $acbd$ and $dbac$
and so on. The idea is that for all candidates besides the ones under
consideration, preferences for and against them cancel each other
out. Hence the need for two votes for each pair. The
total number of thus-constructed votes is twice the cardinality of the
preference pattern. Thus in our example, we obtain a set of votes
which yield exactly the relations in the given
preference pattern.

\subsection{From a {\boldmath $netadv$} Function to a Set of Votes}
We consider the following useful property of the $netadv$ function
when constructing the corresponding election:

\begin{theorem}
For all pairs of candidates $c_{i}, c_{j}$ where $c_{i} \neq c_{j}$,
the values of  $netadv(c_{i}, c_{j})$ are either all even or all odd.
\end{theorem}
\begin{proof}

Consider two blocs of votes where each bloc takes one side
in a pairwise election.  Let the sizes of the blocs be $x$ and $y$
such that $x + y = n$.

\noindent
If n is even: 
\begin{itemize}
\item Then $x,y$ are either both odd or either both even since the sum
  components of an even number are either both even, or both odd.
\item The difference between two even numbers or two odd numbers is
  always even.
\end{itemize}
If n is odd:
\begin{itemize}
\item Then $x,y$ are either odd and even, or even and odd,
  respectively, since the sum components of an odd number are always a
  combination of even and odd.
\item The difference between an even and odd number is always odd.
\end{itemize}
Thus, all the $netadv$ values are either all even or all odd.
\end{proof}

We can see that the $netadv$ function corresponds to elements of a
preference pattern: $netadv(c_{i}, c_{j}) > 0$ corresponds to
$c_{i}Pc_{j}$, $netadv(c_{i}$, $c_{j}) = 0$ corresponds to
$c_{i}Ic_{j}$, and  $netadv(c_{i}$, $c_{j}) < 0$ corresponds to
$c_{j}Pc_{i}$.  However, the key difference between the $netadv$ functions
and preference-pattern elements is that $netadv$ has
scores, which we must factor into our construction.

Since we construct two votes for each pair of candidates, the problem
of applying McGarvey's method to a $netadv$ relation with an even
score is trivial---for each of the two preference orderings
constructed we simply have $n/2$ such votes.  In fact, given a
$netadv$ function, if one $netadv$ score is even, then (1) the number
of votes will be even, and (2) every $netadv$ score in that set will
be even.  The converse also applies:
if one $netadv$ score in a given set is odd, then the number of votes
will be odd, and every $netadv$ score in the given set will be odd.

Applying McGarvey's method to an odd set of $netadv$ scores requires a
slight tweak. We first must select some arbitrary ordering of the
candidates.  Then for any $netadv$ score $s = netadv(a,b)$ where $a$
precedes $b$ in the ordering, we construct $s-1$ (net) votes (or units
of score) exactly as in the case where the scores are all even (i.e.,
two preference orderings, with $(s-1)/2$ such votes for each one). In
the case of $s = netadv(b,a)$ where $a$ precedes $b$, we will instead
create $(s+1)/2$ pairs of votes to give a score of $s+1$ for $b$ over
$a$.  To obtain the last points, we construct one preference ordering
corresponding to the previously chosen ordering of the
candidates. This will give one net vote to every pair $a,b$ where $a$
precedes $b$, and minus one vote where $b$ precedes $a$.  Thus we
achieve the desired odd numbers for each $netadv$ input.

\paragraph{Number of Constructed Votes}

In the above two cases ($netadv$ scores being even or odd) we can see
the upper bound on the number of votes constructed is one for every
unit of score across all $netadv$ values.  The size of the set of
votes will be bounded by $\sum_{c_i,c_j \in C} |netadv(c_{i},
c_{j})|$.  Thus given a $netadv$ function with bounded value, we can
construct a reasonably small set of votes.





\section{Graph Representation of the {\boldmath $netadv$} Function}
Consider a complete directed antisymmetric graph $G(V_G,E)$ where
$V_G$ is the set of vertices and $E$ is the set of directed weighted
edges. We can trivially see that there will be $m(m-1)/2$ directed
edges, where $m = ||V_G||$.  Recall that this is the same as the
number of distinct scores required to represent a $netadv$
function (where $m$ is the number of candidates).  Thus we can
represent a $netadv$ function with such a graph $G$ where $V_G$ is the
set of candidates $C$ and $E$ encodes the $netadv$ function.
Similarly, given such a graph, we can obtain an equivalent $netadv$
function.

\bibliographystyle{alpha}
\bibliography{grypiotr2006}
\end{document}